\documentclass[reqno]{amsart}
\usepackage{amsmath,amsfonts,amsgen,amstext,amsbsy,amsopn,amsthm, amssymb}

\date{April 27, 2010}

\newtheorem{Theorem}{Theorem}
\newtheorem{proposition}{Proposition}
\newtheorem{Lemma}{Lemma}
\newtheorem{corollary}{Corollary}

\numberwithin{equation}{section}

\newcommand{\R}{\mathbb{R}}

\newcommand{\be}{\begin{equation}}
\newcommand{\ee}{\end{equation}}
\newcommand{\bea}{\begin{align}}
\newcommand{\eea}{\end{align}}

\newcommand\bfu{\mathbf{u}}
\newcommand\bfj{\mathbf{j}}

\DeclareMathOperator{\const}{const\,}

\begin{document}

\title{Stability and Absence of Binding for Multi-Polaron Systems}

\author[R.L. Frank]{Rupert L. Frank} \address{R.L. Frank, Department of Mathematics, Princeton University, Washington Road, Princeton, NJ 08544, USA} \email{rlfrank@math.princeton.edu}

\author[E.H. Lieb]{Elliott H. Lieb}
 \address{E.H. Lieb, Departments of Mathematics and Physics, Princeton University, P.O. Box 708, Princeton, NJ 08544, USA} \email{lieb@princeton.edu}

\author[R. Seiringer]{Robert Seiringer}
 \address{R. Seiringer, Department of Physics, Princeton University, P.O. Box 708, 
  Princeton, NJ 08544, USA} \email{rseiring@princeton.edu}

\author[L.E. Thomas]{Lawrence E. Thomas}
 \address{L.E. Thomas, Department of Mathematics, University of Virginia, Charlottesville, VA 22904, USA} \email{let@virginia.edu}

\begin{abstract}
  We resolve several longstanding problems concerning the stability
  and the absence of multi-particle binding for $N\geq 2$
  polarons. Fr\"ohlich's 1937 polaron model describes non-relativistic
  particles interacting with a scalar quantized field with coupling
  $\sqrt\alpha$, and with each other by Coulomb repulsion of strength
  $U$. We prove the following: (i) While there is a known
  thermodynamic instability for $U<2\alpha$, stability of matter
  \emph{does} hold for $U>2\alpha$, that is, the ground state energy
  per particle has a finite limit as $N\to\infty$. (ii) There is no
  binding of any kind if $U$ exceeds a critical value that depends on
  $\alpha$ but not on $N$. The same results are shown to hold for the
  Pekar-Tomasevich model.

\end{abstract}

\maketitle
\renewcommand{\thefootnote}{${}$} \footnotetext{\copyright\, 2010 by
  the authors. This paper may be reproduced, in its entirety, for
  non-commercial purposes.}

\section{Introduction and main results}

Fr\"{o}hlich's large polaron \cite{Fr} is a model for the motion of an
electron in a polar crystal and it is also relevant as a simple model of
non-relativistic quantum field theory.  Consequently there is a huge
literature, both experimental and theoretical, devoted to its
study. See, e.g., \cite{GeLo,Sp,fomin,Mo, AlDe} and references
therein. Our concern here is with the binding or non-binding of
several polarons: whether the ordinary
Coulomb repulsion among the electrons can, if strong enough, prevent
the binding that would otherwise be created by the electric field of
the polar crystal. We are also interested in the \emph{stability of
  matter}, i.e., whether the energy of $N$ polarons is bounded below
by a constant times $N$ even when there is binding.

In this model the single polaron, which is one non-relativistic electron interacting with a phonon field, has the Hamiltonian
\begin{equation}
 \label{eq:ham}
H^{(1)} = p^2 - \sqrt\alpha \phi(x) +H_f \,.
\end{equation}
This Hamiltonian acts in the Hilbert space $L^2(\R^3)\otimes \mathcal F$, where $\mathcal F$
is the bosonic Fock space for the longitudinal optical modes of the
crystal,  with scalar creation and annihilation operators $a^\dagger(k)$
and $a(k)$ satisfying $[a(k),a^\dagger(k')]=\delta(k-k')$. The electron momentum is $p=-i\nabla$, the phonon field energy is
\begin{equation}\label{Hf.eq}
H_f = \int_{\R^3} dk\, a^\dagger(k) a(k) \,,
\end{equation}
and the interaction of the crystal modes with the electron is
\begin{equation}
\phi(x) = \frac{1}{\sqrt2\pi} \int_{\R^3} \frac{dk}{|k|} \left( e^{ikx} a(k)
  + e^{-ikx} a^\dagger(k) \right) \,,
\end{equation}
with coupling constant $\alpha>0$. (In another frequently used
convention $\alpha$ is replaced by $\alpha/\sqrt2$.) The ground state
energy $E^{(1)}(\alpha)$ is the infimum of the spectrum of $H^{(1)}$.
Because of translation invariance,  $E^{(1)}(\alpha)$ cannot be expected to be an
eigenvalue,  and indeed it is not; this was proved in \cite{GeLo1} using methods developed in \cite{Fr2}.

A noteworthy feature of the phonon field energy $H_f$ is its flat
dispersion relation, i.e., there is no non-constant function
$\omega(k)$ in the integrand of (\ref{Hf.eq}).  The energy of infrared
phonons does not go to zero as $|k|\to 0$, while the ultraviolet
energy is finite when $|k|\to\infty$.  If we tried to minimize the energy in a na\"ive way by
completing the square, we would end up with a Coulomb-like $\int dk\,
|k|^{-2}$ self-energy of the polaron. The non-integrability for large
$k$ would lead to a divergent self-energy, but this divergence is
actually mitigated by the electron kinetic energy $p^2$ and the
uncertainty principle. While the single polaron has finite energy,
another problem remains for the many-polaron system; the energy is
finite, but stability of matter will not hold unless a sufficiently
strong Coulomb repulsion among the electrons is included in the
Hamiltonian.

The Hamiltonian for $N$ electrons is
\begin{equation}
H^{(N)}_U = \sum_{i=1}^N \left(p_i^2 - \sqrt\alpha\phi(x_i) \right) +
H_f + U\, V_C(X)
\end{equation}
with $X=(x_1,\dots,x_N)\in \R^{3N}$ and 
\begin{equation}
V_C(X) = \sum_{i<j} \frac{1}{|x_i-x_j|} \,.
\end{equation}
We impose no symmetry restrictions on the electrons, which means that
our lower bounds apply equally to bosons or fermions or particles with
no symmetry restrictions (boltzons).  The Hilbert space is then
$L^2(\R^{3N})\otimes \mathcal F$. Particle spin is irrelevant for our
results and is ignored.  Physically, the parameter $U$ is the square
of the electron charge, and it satisfies $U>2\alpha$
\cite{Fr}. Nevertheless,  we will consider all values $U\geq 0$.

The ground state energy of $H^{(N)}_U$ is denoted by
$E^{(N)}_U(\alpha)$, and the \textit{binding energy} is defined to be
$\Delta E^{(N)}_U(\alpha) = N E^{(1)}(\alpha) - E^{(N)}_U(\alpha)$.
We will prove three theorems about these quantities. 
They were previously summarized in an announcement \cite{FLST}.  Our main goal is to find
conditions on $U$ and $\alpha$ such that no binding occurs, i.e.,
$\Delta E^{(N)}_U(\alpha)=0$. Of particular physical interest is the
case $N=2$ (\textit{bipolaron}). 

\begin{Theorem}[\textbf{Absence of binding for {\mathversion{bold}$N$} polarons}]\label{thm:nobindingN}
 For given $\alpha>0$ there is a finite $U_c(\alpha)>2\alpha$ such that
\begin{equation}
 \label{eq:bindabs}
\Delta E^{(N)}_U(\alpha) =0
\quad\text{for all}\ N\geq 2
\end{equation}
whenever $U\geq U_c(\alpha)$.
\end{Theorem}

Our proof is constructive and gives an explicit upper bound on
$U_c(\alpha)$; see the discussion at the end of
Section~\ref{sec:nobindn}. This bound on $U_c(\alpha)$ is
linear in $\alpha$ for large $\alpha$, which is the correct
behavior. Presumably, the true $U_c(\alpha)$ behaves linearly even for
small $\alpha$, but this remains an \textit{open problem}.

For the bipolaron, $N=2$, our proof is simpler
and yields the sharper result that the critical $U_c(\alpha)$ indeed
obeys a linear law:

\begin{Theorem}[\textbf{Absence of binding for bipolarons}]\label{thm2}
 Let $N=2$. For some constant  $C<26.6$,
\begin{equation}
 \label{eq:bipol}
\Delta E^{(2)}_U(\alpha) =0
\end{equation}
whenever $U\geq 2C\alpha$.
\end{Theorem}

The optimal constant $C$ in Theorem~\ref{thm2} is presumably
much closer to $1$ than the bound we derive. It is not equal to $1$,
however. In the strong coupling limit $\alpha\to \infty$, binding
occurs if $U\leq 2.3\, \alpha$ \cite{VSPD,fomin}. For small $\alpha$, on the other hand, variational calculations \cite{VSPD,BrKlDe} suggest that bipolaron binding does not occur at all for any $U\geq 2\alpha$. The proof of this remains an \emph{open problem}.

On the other hand, if binding does occur, we would like to know how
the binding energy depends on $N$; in particular, is there
\emph{stability of matter}, in the sense that $\Delta
E^{(N)}_U(\alpha) \leq C(U,\alpha) N$ for all $N$?

This linear  bound, if it exists, implies the existence of the
\emph{thermodynamic limit}
$$
\lim_{N\to\infty} N^{-1} E^{(N)}_U(\alpha) \,.
$$
The proof is a simple consequence of the sub-additivity of the energy, i.e., 
\begin{equation}\label{eq:subadd}
E^{(N+M)}_U(\alpha) \leq E^{(N)}_U(\alpha)+ E^{(M)}_U(\alpha)\,,
\end{equation}
which follows from the fact that one can construct variational
functions in which $N$ electrons are localized on the earth and $M$
behind the moon \cite{GrMo} \cite[Sec. 14.2]{LiSe}.

The proof of the linear lower bound is far from obvious; indeed, it is
not always true! Griesemer and M\o ller \cite{GrMo} recently proved
that when $U<2\alpha$, there are positive constants (depending on $U$
and $\alpha$) such that $-c_1 N^{7/3}\geq E^{(N)}_U(\alpha)\geq -c_2
N^{7/3}$. This result holds for particles, like electrons, that
satisfy Fermi statistics. If the electrons were bosons, the result
would be even worse, $-N^3$, as a similar analysis shows. This is the
same behavior as that of gravitating particles in stars \cite[Ch.~13]{LiSe}.
In the opposite regime, $U>2\alpha$, \cite{GrMo} shows that
$E^{(N)}_U(\alpha)\geq -C(U,\alpha) N^2$, independently of statistics, where $C(U,\alpha)\to\infty$
as $U\searrow2\alpha$.  (In their convention the dividing line is
$U=\sqrt2 \alpha$.)

Judging from the physics of the model,  it is reasonable to suppose that
there is a linear law as soon as $U>2\alpha$. This we are able to prove.

\begin{Theorem}[\textbf{Stability for {\mathversion{bold} $U>2\alpha$}}] \label{thm:stabalpha}
 For given $U>2\alpha>0$, $N^{-1} E^{(N)}_U(\alpha)$ is bounded independently of $N$.
\end{Theorem}

Our lower bound on $N^{-1} E^{(N)}_U(\alpha)$ goes to $-\infty$ as
$U\searrow2\alpha$, but we are not claiming that this reflects the
true state of affairs. Whether $\lim_{N\to\infty} N^{-1}
E^{(N)}_{2\alpha}(\alpha)$ is finite or not remains an \textit{open
  problem}; see, however, the discussion of the strong coupling limit
below.

For $U$ in the range $2\alpha<U<U_c(\alpha)$, there are bound states
of an undetermined nature. Does the system become a gas of bipolarons or does it coalesce into a true $N$-particle bound state? If the latter, does this state exhibit a periodic structure, thereby forming a super-crystal on top of the underlying lattice of atoms? This is perhaps the physically most interesting \textit{open problem}.


\subsection{The strong coupling limit}

There is a non-linear differential-integral variational principle associated with the polaron problem,  which gives the exact ground state energy in the limit $\alpha\to \infty$. This variational problem was investigated in detail by Pekar \cite{Pe}. 
Pekar and Tomasevich (PT) \cite{PeTo} generalized it to the bipolaron,  and the extension to $N$-polarons obviously follows from \cite{PeTo}.

The PT functional is the result of a variational calculation and therefore gives an upper bound to the ground state energy $E^{(N)}_U(\alpha)$. In order to compute $\langle\Psi,H^{(N)}_U\Psi\rangle$, one takes a $\Psi$ of the form $\psi\otimes \Phi$ where $\psi\in L^2(\R^{3N})$, $\Phi\in\mathcal F$,  and both $\psi$ and $\Phi$ are normalized. For a given $\psi$ it is easy to compute the optimum $\Phi$, and one ends up with the functional
\begin{align}\nonumber
\mathcal P^{(N)}_{U}[\psi] & := \sum_{i=1}^N \int_{\R^{3N}} |\nabla_i \psi|^2 \,dX + U \sum_{i<j} \int_{\R^{3N}} \frac{|\psi(X)|^2}{|x_i-x_j|} \,dX \\ & \quad - \alpha \iint_{\R^3\times\R^3} \frac{\rho_\psi(x)\, \rho_\psi(y)}{|x-y|} \,dx\,dy \,, \label{eq:pekar}
\end{align}
where $dX=\prod_{k=1}^N dx_k$, and
\begin{equation}
\rho_\psi(x) = \sum_{i=1}^N \int_{\R^{3(N-1)}} |\psi(x_1,\ldots,x,\ldots,x_N)|^2 \,dx_1\cdots \widehat{dx_i} \cdots dx_N
\end{equation}
with $x$ at the $i$-th position, and $\widehat{dx_i}$ meaning that $dx_i$ has to be omitted in the product $\prod_{k=1}^N dx_k$. The ground state energy is
\begin{equation}
\mathcal E^{(N)}_U(\alpha) = \inf\left\{ \mathcal P^{(N)}_{U}[\psi] : \int_{\R^{3N}} |\psi|^2\,dX = 1 \right\} \,.
\end{equation}
Hence the variational argument above gives the upper bound
\begin{equation}
 E^{(N)}_U(\alpha) \leq \mathcal E^{(N)}_U(\alpha) = \mathcal E^{(N)}_{U/\alpha}(1) \ \alpha^2 \,.
\end{equation}
(The equality follows by scaling.)
For $N=1$ this upper bound is due to Pekar; numerically, one has $\mathcal E^{(1)}(\alpha) \approx -(0.109) \alpha^2$ \cite{Mi}. Moreover, the minimization problem for $\mathcal E^{(1)}(\alpha)$ has a unique minimizer (up to translations), see \cite{Li}.

The upper bound for $N=1$ was widely understood to be asymptotically exact for large $\alpha$. A proof of this was finally achieved by Donsker and Varadhan \cite{DoVa}, using large deviation theory applied to the functional integral discussed below. Later, this fact was rederived in  \cite{LiTh}  by operator methods,  and it was shown that the error was no worse than $\alpha^{9/5}$ for large $\alpha$.

The fact that for fixed ratio $\nu=U/\alpha\geq 0$ 
\begin{equation}
 \label{eq:pekarlimit}
\lim_{\alpha\to\infty} \alpha^{-2} E^{(N)}_{U}(\alpha) = \mathcal E^{(N)}_{\nu}(\alpha=1)
\end{equation}
for $N=2$ and any $\nu\geq 0$ was first noted in \cite{spohn}. This is also valid for arbitrary $N$.

It follows from the limiting relation \eqref{eq:pekarlimit}, together with the fact that our bound on $U_c(\alpha)$ is linear for large $\alpha$, that our three theorems about $E^{(N)}_U(\alpha)$ transfer to the same theorems about $\mathcal E^{(N)}_U(\alpha)$ which we state next.

\begin{corollary}[{\bf Stability and absence of binding for the PT functional}]\label{pekarcor}$\quad$
1. Stability holds for $U>2\alpha$, that is, for any $\nu>2$ there is a constant $C(\nu)$ such that
\begin{equation}
\label{eq:stabpekar}
 \mathcal E^{(N)}_U(\alpha) \geq -C(\nu) \alpha^2 N 
\quad\text{for all}\ N\geq 2
\end{equation}
whenever $U=\nu\alpha>0$. \\
2. There is no binding if $U/\alpha$ is large enough, that is, there is a finite $\nu_c>2$ such that
\begin{equation}
 \label{eq:bindabspekar}
\mathcal E^{(N)}_U(\alpha) = N \mathcal E^{(1)}(\alpha)
\quad\text{for all}\ N\geq 2
\end{equation}
whenever $U\geq \nu_c \alpha$.
\end{corollary}

We note in passing that the critical $U$ in the PT model depends linearly on $\alpha$. This follows by scaling.

In Section \ref{sec:pekar} we will give a direct and easier proof of Corollary \ref{pekarcor} that has no need of the functional integral machinery and leads to better constants. In particular, with our new proof we find that
\begin{equation}
 \label{eq:stabpekarconst}
C(\nu) \leq
\begin{cases}
 (0.0280) \nu^3/(\nu-2) & \quad\text{if}\ 2<\nu<3 \,,\\
0.755 & \quad\text{if}\ \nu\geq 3 \,.
\end{cases}
\end{equation}
For fermions, a stronger result than Corollary \ref{pekarcor} was obtained in \cite{GrMo}. They prove that at critical coupling $U=2\alpha$ one has
\begin{equation}\label{eq:thm4}
\mathcal E^{(N)}_{2\alpha}(\alpha,q) \geq - (0.461)^2 \alpha^2 q^{2/3} N
\quad\text{for all}\ N\geq 2
\ \text{and all}\ 1\leq q\leq N \,,
\end{equation}
where $\mathcal E^{(N)}_{U}(\alpha,q)$ is the infimum of $\mathcal P^{(N)}_{U}[\psi]$ restricted to fermions with $q$ spin states ($q=2$ for electrons). Recall that the single polaron energy is $\mathcal E^{(1)}(\alpha) = -(0.109)\alpha^2$ \cite{Mi}. For completeness we repeat the short proof of (\ref{eq:thm4}) in Section~\ref{sec:pekar}.


\subsection{Previous results on the polaron ground state energy}

In addition to the large $\alpha$ asymptotics just mentioned, there have been several other rigorous results on the ground state energy of the (single) polaron, some of which we will use here.

\begin{itemize}
 \item[(i)]
One of the earliest results was the variational calculation of $E^{(1)}(\alpha)$ for small $\alpha$ by Gurari and by Lee, Low, and Pines \cite{Gu,LePi,LeLoPi} which leads to
\begin{equation}
 \label{eq:llp}
E^{(1)}(\alpha) \leq -\alpha 
\qquad \text{for all}\ \alpha \,.
\end{equation}
\item[(ii)]
A lower bound, which validates the conclusion that $E^{(1)}(\alpha) \sim -\alpha$ as $\alpha\to 0$, was obtained in \cite{LiYa}. They prove that
\begin{equation}
 \label{eq:ly}
E^{(1)}(\alpha) \geq -\alpha -\frac13 \alpha^2 
\qquad \text{for all}\ \alpha \,.
\end{equation}
(To derive \eqref{eq:ly} use $p=1+2\alpha/3$ in \cite[Eq. (24)]{LiYa}.) We note the fact that \eqref{eq:ly} has the correct power law behavior for both small and large $\alpha$.
\item[(iii)]
\textit{The functional integral formulation}: The large time behavior of the heat kernel $\exp(-TH^{(N)}_U)$ gives us the ground state energy as
\begin{equation}
E^{(N)}_U (\alpha)= -\lim_{T\to\infty} T^{-1} \ln \left\langle\exp\big(-TH^{(N)}_U\big)\right\rangle \,,
\end{equation}
where $\langle \,\cdot\, \rangle$ denotes the expectation in a suitable state, i.e., a normalized vector in the Hilbert space. Since the phonon operators can be realized as the coordinates of a quantum-mechanical harmonic oscillator (one for each value of $k$), we can apply the Feynman-Kac formula for the evaluation of $\langle\exp(-TH^{(N)}_U)\rangle$. Since the harmonic oscillator coordinates appear only linearly and quadratically in the exponent, they can then be integrated out explicitly.  One obtains a formula due to Feynman \cite{Feynman55}, see also \cite[Sec. 5.3]{Ro} for a careful discussion. The conclusion is that
\begin{equation}
 \label{eq:pathint}
E^{(N)}_U(\alpha) = - \lim_{R\to\infty} \lim_{T\to\infty} T^{-1} \ln Z^{(N)}_{U,R}(T) \,,
\end{equation}
where
\begin{align}
\label{eq:pathint2}
 Z^{(N)}_{U,R}(T) := \int_{B_R} dx_1 \cdots \int_{B_R} dx_N & \int dW^T_{x_1}(\omega_1) \cdots dW^T_{x_N}(\omega_N) 
\chi_{B_R}(\omega_1) \cdots \chi_{B_R}(\omega_N) \notag \\
\times & \exp\left( \alpha \int_\R \frac{ds\, e^{-|s|}}{2} 
\sum_{i,j=1}^N  \int_0^T  \frac {dt}{|\omega_i(t)-\omega_j(t+s)|} \right)  \notag \\
\times & \exp\left( - U \sum_{i<j}  \int_0^T  \frac {dt}{|\omega_i(t)-\omega_j(t)|} \right) \,.
\end{align}
Here $dW^T_x$ denotes the Wiener measure of closed Brownian paths in $\R^3$ with period $T$ starting and ending at $x$. Moreover, $B_R$ denotes the ball centered at the origin of radius $R$, and the characteristic function of the path $\chi_{B_R}(\omega_j)$ is 1 if $\omega_j$ stays inside the ball $B_R$ for all times, and zero otherwise. The argument of $\omega_j(t+s)$ is understood modulo $T$.
\end{itemize}


\section{The two-polaron problem: Absence of Binding}

We first consider the special case $N=2$ and prove Theorem \ref{thm2}. It is convenient to structure the proof in three steps. 

\emph{Step 1. Partition of the interparticle distance.} We  choose a quadratic partition of unity (IMS localization \cite[Thm~3.2]{IMS}) and localize the particles according to their \emph{relative distance}. (In the $N$-particle case later on, we will localize with respect to the nearest neighbor distance, which, for $N=2$, is the same as the relative distance.) This kind of localization is one of the principal novel features of our analysis.

In order to construct this partition, we pick some parameters $b>1$ and $\ell > 0$, and let
\begin{equation}\label{phi1}
\varphi(t) := \left\{ 
\begin{array}{ll} 
0 & \text{for $t\leq \ell/b$} \,, \\
\sin\frac \pi 2 \frac{t-\ell/b}{\ell-\ell/b} & \text{for $\ell/b\leq t \leq \ell $} \,,\\
\cos\frac \pi 2\frac{t-\ell}{b\ell-\ell} & \text{for $\ell\leq t\leq b\ell$} \,,\\
0 & \text{for $t\geq b\ell$} \,. 
\end{array} \right.
\end{equation}
For $j\geq 1$, let $\varphi_j(t) := \varphi(b^{1-j}t)$, and for $j=0$, let 
\begin{equation}\label{phi2}
\varphi_0(t) := \left\{ \begin{array}{ll} 1 & \text{for $t\leq \ell/b$} \,, \\ \cos\frac \pi 2 \frac{t-\ell/b}{\ell-\ell/b} & \text{for $\ell/b\leq t\leq \ell$} \,, \\ 0 & \text{for $t\geq \ell$}\,. \end{array}  \right.
\end{equation}
Then
\begin{equation}
\sum_{j\geq 0} \varphi_j(t)^2 = 1 \quad \text{for all $t\geq 0$} \,.
\end{equation}

Using the IMS localization formula, we can write, for any wave function $\psi$,
\begin{equation}\label{2.4}
\langle \psi | H^{(2)}_U |\psi \rangle  
= \sum_{j\geq 0} \left\langle \psi_j \left| H^{(2)}_U - 2 \sum_{k\geq 0} \left| \varphi_k'(|x_1-x_2|)\right|^2 \right|  \psi_j \right\rangle
 =: \sum_{j\geq 0} e_j \|\psi_j\|^2
\end{equation}
with $\psi_j(x_1,x_2) = \psi(x_1,x_2) \varphi_j(|x_1-x_2|)$ and with numbers $e_j$ (depending on $\psi_j$). Our goal is to prove that $e_j\geq 2 E^{(1)}(\alpha)$ for all $j$ if $U\geq 2C\alpha$. If this is indeed the
case, then the right side of (\ref{2.4}) exceeds $2E^{(1)}(\alpha)\sum_j\|\psi_j\|^2
= 2E^{(1)}(\alpha)\|\psi\|^2$, which is the assertion of the theorem. 

For our bounds we shall use the fact that on the support of $\varphi_j(|x_1-x_2|)$, the localization error is dominated by
\begin{equation}\label{eq:locerror}
 \sum_{k\geq 0} \left| \varphi_k'(|x_1-x_2|)\right|^2 \leq \frac{\pi^2}{4(\ell-\ell/b)^2} \times \left\{ \begin{array}{ll} 1 & \text{if $j=0$} \,, \\ b^{2(1-j)} & \text{if $j\geq 1$} \,. \end{array}\right.
\end{equation}
Moreover, on these supports, we shall bound the Coulomb repulsion from below by
\begin{equation}\label{boundu}
\frac U{|x_1-x_2|} \geq  b^{-j} \frac U\ell \qquad\text{for all}\ j\geq 0 \,.
\end{equation}
It is clear from \eqref{eq:locerror} and \eqref{boundu} that by choosing $U$ large enough, we can dominate the negative localization error by a part of the positive Coulomb term. What remains is to dominate the polaronic attraction by the remainder of the Coulomb repulsion. For this, we distinguish between the cases $j\geq 1$ and $j=0$.

\bigskip

\emph{Step 2. The case $j\geq 1$; Energy estimate for separated particles.}
We further localize each of the two particles to its own ball of radius $b^j L$ for some parameter $L>0$. This will entail an additional localization error. Concretely, let 
\begin{equation}\label{defchi}
\chi(x) = \frac 1{\sqrt{2\pi} |x|}
\left\{ 
\begin{array}{ll} 
\sin(\pi|x|) & \text{for $|x|\leq 1$} \,, \\
0 & \text{for $|x|\geq 1$} \,,
\end{array}\right.
\end{equation}
and note that $\int dx \,\chi(x)^2 = 1$ and $\int dx |\nabla \chi(x)|^2 = \pi^2$. With
\begin{equation}
\psi_{j,u_1,u_2}(x_1,x_2) = \psi_j(x_1,x_2) (b^j L)^{-3} \chi(b^{-j}(x_1-u_1)/L) \chi(b^{-j}(x_2-u_2)/L)
\end{equation}
we have, by a continuous version of the IMS localization formula, 
\begin{align}\nonumber\
&  \left\langle \psi_j \left| H^{(2)}_U - 2 \sum_{k\geq 0} \left| \varphi_k'(|x_1-x_2|)\right|^2 \right|  \psi_j \right\rangle \\ \nonumber & = \int_{\R^3} du_1 \int_{\R^3} du_2 \left\langle \psi_{j,u_1,u_2} \left| H^{(2)}_U  - 2\sum_{k\geq 0} \left| \varphi_k'(|x_1-x_2|)\right|^2 - \frac {2 \|\nabla\chi\|^2}{b^{2j} L^2} \right|  \psi_{j,u_1,u_2}\right\rangle \\ & \geq \int_{\R^3} du_1 \int_{\R^3} du_2 \left\langle \psi_{j,u_1,u_2} \left| H^{(2)}_U  - b^{-2j} \left( \frac {b^2 \pi^2}{2(\ell-\ell/b)^2} + \frac {2\pi^2}{L^2}\right) \right|  \psi_{j,u_1,u_2}\right\rangle \,. \label{48}
\end{align}
The latter inequality comes from \eqref{eq:locerror}.
 Note that since $|x_1-x_2|\geq b^{j-2} \ell$ on the support of $\varphi_j$,  the wave function $\psi_{j,u_1,u_2}$ is non-zero only if the two balls of radius $b^j L$ centered at $u_1$ and $u_2$, respectively,  are separated at least a distance
\begin{equation}\label{bounddd}
d \geq b^{j-2}\ell-4 b^jL \,.
\end{equation}
We choose this to be positive by requiring that $L< \ell/(4b^2)$.

\begin{Lemma}\label{lem3}
Assume that $\psi$ is normalized and supported in $B_1\times B_2$ where $B_1$ and $B_2$ are disjoint balls of some radius $R$, separated a distance $d$. Then
\begin{equation}\label{ls}
\langle \psi| H^{(2)}_{0}|\psi\rangle \geq 2 E^{(1)}(\alpha) - \frac {2\alpha} d 
\end{equation}
\end{Lemma}

This lemma will be proved in Subsection \ref{sec:pathintproofs}. It is an easy consequence of the functional integral representation of the ground state energy.

We apply inequality \eqref{ls} to (\ref{48}). Using the bounds (\ref{bounddd}) on $d$ and (\ref{boundu}) on the Coulomb potential, we conclude that $e_j$ is bounded from below as
\begin{equation}
e_j \geq 2 E^{(1)}(\alpha) - b^{-j} \frac{2\alpha}{ \ell/b^2-4 L} + b^{-j} \frac U\ell - b^{-2j} \left( \frac {b^2 \pi^2}{2(\ell-\ell/b)^2} + \frac {2\pi^2}{L^2}\right) \,.
\end{equation}
This last expression is $\geq 2 E^{(1)}(\alpha)$ for all $j\geq 1$ if and only if
\begin{equation}\label{cu2}
\boxed{ \  U \geq \frac{2 \alpha \ell }{ \ell/b^2-4 L}  + \frac {b \ell \pi^2}{2(\ell-\ell/b)^2} + \frac {2\pi^2 \ell}{L^2 b} \ } \qquad\text{(the $j\geq 1$ condition)} \,.
\end{equation}

\bigskip

\emph{Step 3. The case $j=0$; Energy estimate for neighboring particles.}
Because of \eqref{eq:locerror} and \eqref{boundu} we have the lower bound
\begin{equation}\label{eq:e0}
e_0 \geq E_0^{(2)}(\alpha) - \frac{\pi^2}{2(\ell-\ell/b)^2}+ \frac U\ell \,,
\end{equation}
where $E^{(2)}_0(\alpha)$ denotes the two-polaron energy in the absence of Coulomb repulsion, i.e., for $U=0$. The following lemma compares this energy with $2 E^{(1)}(\alpha)$.

\begin{Lemma}\label{lem2} For all $\alpha>0$,
\begin{equation}\label{eq:lem2}
E^{(2)}_0(\alpha) \geq 2 E^{(1)}(\alpha) - \frac 7{3} \alpha^2 \,.
\end{equation}
\end{Lemma}

Also this lemma uses the path integral formulation and we defer the proof to Subsection \ref{sec:pathintproofs}. At this point we will utilize it to conclude the proof of Theorem \ref{thm2}. The constant $7/3$ in \eqref{eq:lem2} is certainly not optimal, and an improvement would lead to a better constant in Theorem \ref{thm2}.

It follows from \eqref{eq:e0} and \eqref{eq:lem2} that $e_0\geq 2 E^{(1)}(\alpha)$ if 
\begin{equation}\label{cu1}
\boxed{\ U \geq \frac 7{3} \ell \alpha^2  +  \frac{\pi^2 \ell}{2(\ell-\ell/b)^2} \ }
\qquad\text{(the $j=0$ condition)} \,.
\end{equation}
Numerical evaluation shows that the two conditions (\ref{cu2}) and (\ref{cu1}) on $U$ can be satisfied for an appropriate choice of $b$, $\ell$, and $L$ if $U \geq 61\, \alpha$. (Choose $b=1.2$, $\ell= 22.8\, \alpha^{-1}$ and $L = 0.142 \, \ell$.) For $U$ satisfying these conditions, each $e_j\geq 2E^{(1)}(\alpha)$. This completes the proof of Theorem~\ref{thm2} with the bound on the constant $C< 30.5$.
\qed

\bigskip

In order to improve the bound on the constant $C$ of Theorem~\ref{thm2}, we replace Lemma~\ref{lem3} by the following alternative bound.

\begin{Lemma}\label{lem4}
Under the same assumptions as in Lemma~\ref{lem3},
\begin{equation}\label{imp_bound}
\langle \psi| H^{(2)}_0|\psi\rangle \geq 2 E^{(1)}(\alpha) - \left\langle \psi \left|  \frac{2 \alpha}{|x_1-x_2|} \right|\psi\right\rangle - \frac {16\alpha R}{\pi^2 d(d+4R)} \,.
\end{equation}
\end{Lemma}


 In the appendix we provide a proof 
 of this lemma using the Rayleigh-Ritz variational principle. The proof is certainly not easier than the one of Lemma~\ref{lem3} using
the functional integral method, but it may be of use in other
applications where a functional integral approach is not as convenient (or available).  The method
 does point up the utility of  localizing the phonon field about
the respective particles -- in this case, to half-spaces each containing a particle.

 We apply the bound (\ref{imp_bound}) to (\ref{48}), with $R= b^{j} L$ and $d$ satisfying (\ref{bounddd}), and conclude that
\begin{equation}
e_j \geq 2 E^{(1)}(\alpha) + b^{-j} \frac{U-2\alpha}\ell - b^{-j}\frac {16\alpha  L b^2}{\ell \pi^2 (\ell/b^2-4 L) } - b^{-2j} \left( \frac {b^2 \pi^2}{2(\ell-\ell/b)^2} + \frac {2\pi^2}{L^2}\right)\,.
\end{equation}
This expression is $\geq 2 E^{(1)}(\alpha)$ for all $j\geq 1$ if and only if  
\begin{equation}\label{cu3}
\boxed{\  U \geq 2\alpha + \frac {16\alpha  L b^2 }{ \pi^2 (\ell/b^2-4 L) }  + \frac {b \ell \pi^2}{2(\ell-\ell/b)^2} + \frac {2\pi^2 \ell}{L^2 b}\ }
\end{equation}
Eqs.~(\ref{cu1}) and (\ref{cu3}) are satisfied for $U \geq  53.2\, \alpha$ with the choice  $b=1.23$, $\ell= 19.7\, \alpha^{-1}$ and $L = 0.15 \, \ell$.


\subsection{Some uses of the path integral}\label{sec:pathintproofs}

Lemmas \ref{lem3} and \ref{lem2}, used in the previous subsection, will be proved here.

\begin{proof}[Proof of Lemma \ref{lem3}]
  We use a Feynman-Kac representation similar to \eqref{eq:pathint}. It implies that the infimum of the left side of (\ref{ls}) over all $\psi$ with the required support properties equals
\begin{equation}
 -\lim_{T\to\infty} \frac 1 T \ln \, Z_{B_1,B_2}(T) 
\end{equation}
where
\begin{align}\nonumber
Z_{B_1,B_2}(T) & := \int_{B_1} dx_1 \int_{B_2} dx_2 \int dW_{x_1}^T(\omega_1) dW_{x_1}^{T}(\omega_2) \chi_{B_1}(\omega_1)  \chi_{B_2}(\omega_2) \\ 
& \qquad \times \exp\left( \alpha \int_\R \frac{ds\, e^{-|s|}}{2} 
\sum_{i,j=1}^2  \int_0^T  \frac {dt}{|\omega_i(t)-\omega_j(t+s)|} \right)\,.
\end{align}
Here $dW_{x_j}^T$ denotes the Wiener measure of closed Brownian paths in
$\R^3$ with period $T$ starting and ending at $x_j$, and $\chi_{B_j}(\omega_j)$ is 1 if $\omega_j$
stays inside the ball $B_j$ for all times, and zero otherwise. 
Since
$|\omega_1(t)-\omega_2(t+s)|\geq d$ for all $t$ and $s$ we see that $Z_{B_1,B_2}(T)$ is bounded from above by 
$$
e^{2\alpha T/d} \prod_{j=1}^2 \left( \!\int_{B_j} \!\!dx \!\!\int \!\!dW_x^T(\omega_j) \chi_{B_j}(\omega_j) \exp\left( \alpha \int_\R \!\frac{ds\, e^{-|s|}}{2} \int_0^T \!\!\frac {dt}{|\omega_j(t)-\omega_j(t+s)|} \right)    \right).
$$
Replacing $\chi_{B_j}(\omega_j)$ by its upper bound $1$, we deduce inequality \eqref{ls}.
\end{proof}

\begin{proof}[Proof of Lemma \ref{lem2}]
Application of the Cauchy-Schwarz inequality in the path integral \eqref{eq:pathint2} yields
\begin{align}\nonumber
Z^{(2)}_{0,R}(T)^2 & \leq  \iint dW^T_R(\omega_1) dW^{T}_R(\omega_2) \\
\nonumber & \qquad\quad \times \exp\left( 2 \alpha \int_\R  \frac{ds\, e^{-|s|}}{2} \sum_{i=1}^2  \int_0^T  \frac {dt }{|\omega_i(t)-\omega_i(t + s)|}\right) \\
\nonumber & \quad  \times \iint dW_R^T(\omega_1) dW_R^{T}(\omega_2) \\
& \qquad\quad  \times \exp\left( 4 \alpha  \int_\R \frac{ds\, e^{-|s|}}{2}  \int_0^T  \frac {dt}{|\omega_1(t)-\omega_2(t + s)|}\right) \,,
\end{align}
where $\int dW_R^T(\omega)$ is short for $\int_{B_R} dx \int dW_x^T(\omega)\chi_{B_R}(\omega)$.
The first factor on the right side equals the square of 
\begin{equation}
 \int dW_R^T(\omega) \exp\left( 2 \alpha \int_\R \frac{ds\, e^{-|s|}}{2} \int_0^T   \frac {dt}{|\omega(t)-\omega(t+ s)|}\right) \,,
\end{equation}
which, in turn, is the one-polaron expression with $\alpha$ replaced
by $2\alpha$.  Using Jensen's inequality, we bound the second factor
from above by
\begin{align}\label{30}
  \int_\R \frac{ds\, e^{-|s|}}{2} \iint dW_R^T(\omega_1)
  dW_R^{T}(\omega_2) \exp\left( 4 \alpha \int_0^T \frac
    {dt}{|\omega_1(t)-\omega_2(t + s)|}\right) \,.
\end{align}
Since closed Brownian paths are invariant under time reparametrization,
the latter integral does not actually depend on $s$, and hence
(\ref{30}) equals
\begin{equation}
\iint dW_R^T(\omega_1) dW_R^{T}(\omega_2) \exp\left( 4 \alpha    \int_0^T  \frac {dt}{|\omega_1(t)-\omega_2(t)|}\right) \,. 
\end{equation}
This functional integral represents two particles in the ball $B_R$ interacting via an attractive Cou\-lomb
potential $-4\alpha/|x_1-x_2|$. This is like the positronium
Hamiltonian whose ground state energy equals
$-2\alpha^2$ in the limit $R\to\infty$. Summarizing, after taking the $T\to\infty$ limit we find
that
\begin{equation}\label{eq:jensenbound}
E^{(2)}_0(\alpha) \geq E^{(1)}(2\alpha) - \alpha^2 \,.
\end{equation}
To finish the proof, we use the bounds  \eqref{eq:llp} and \eqref{eq:ly},
which imply that 
\begin{equation}
E^{(1)}(2\alpha) \geq 2 E^{(1)}(\alpha) -  \frac 43 \alpha^2 \ .
\end{equation}
This, together with \eqref{eq:jensenbound}, proves \eqref{eq:lem2}.
\end{proof}


\section{The $N$-polaron problem: Thermodynamic Stability}

We now consider the case of general $N$ and prove Theorem \ref{thm:stabalpha}. 
We start by localizing particles in balls in order to reduce the problem to
  a local one. We use the sliding technique introduced in \cite{CLY}
  (see also \cite{LS1}).  Pick an even and real-valued function $\chi$
  with compact support, normalized by $\int \chi^2 = 1$, and
  $\omega>0$ large enough such that the function
\begin{equation}
f(x) = \frac 1{|x|} \left( 1- e^{-\omega|x|} \chi*\chi(x)\right)
\end{equation}
is positive definite. (The symbol $*$ means convolution.) The existence
of such an $\omega$ for smooth enough $\chi$ was shown in
\cite[Lemma~2.1]{CLY}. For any operator-valued function $\rho(x)$,
\begin{equation}\label{eq:posdef}
\iint dx\,dy \left( \sum_{i=1}^N \delta(x-x_i) -\rho(x)^\dagger \right) f(x-y) \left( \sum_{i=1}^N \delta(y-x_i) -\rho(y) \right)\geq 0\,.
\end{equation}
We apply this to 
\begin{equation}
\rho(x) = \frac {1}{ (2\pi)^2\sqrt{2\alpha}} \int dk\, |k|  e^{i k x} a(k) 
\end{equation}
and obtain the bound
\begin{equation}
\sum_{1\leq i<j\leq N} \frac{2\alpha}{|x_i-x_j|} - 
\sqrt{\alpha}\sum_{i=1}^N \phi(x_i) + H_f \geq  - \alpha N \omega + 
2\alpha  \int_{\R^3} dz\, I_\omega(z) \ .
\end{equation}
Here,
\begin{align}\nonumber
I_\omega(z) & :=   \sum_{1\leq i<j\leq N} \chi_z(x_i)\frac{e^{-\omega|x_i-x_j|}}{|x_i-x_j|}\chi_z(x_j)  \\ \nonumber & \quad - \frac 12 \sum_{i=1}^N \chi_z(x_i) \int dy \frac{e^{-\omega|x_i-y|}}{|x_i-y|}\chi_z(y)\left(\rho(y)+\rho(y)^\dagger\right) \\ &  \quad + \frac 12 \iint dx\,dy\, \chi_z(x)\rho(x)^\dagger \frac{e^{-\omega|x-y|}}{|x-y|}\rho(y)\chi_z(y)
\label{eq:iomega}
\end{align}
where we denote $\chi_z(x)=\chi(x-z)$. We also note that 
\begin{equation}
p^2 = \int dz \,p \chi_z^2 p = \int dz\, \chi_z p^2 \chi_z - \int dx\,|\nabla\chi(x)|^2,
\end{equation}
and thus
\begin{equation}\label{hhz}
H^{(N)}_U \geq \int dz\, H_z + \left(U-2\alpha\right) V_C -  \alpha N \omega - \frac N2 \int dx\, |\nabla\chi(x)|^2 
\end{equation}
where
\begin{equation}
H_z  := \frac 12 \sum_{i=1}^N \left(  \chi_z(x_i)p_i^2 \chi_z(x_i) + p_i \chi_z(x_i)^2 p_i\right) + 2\alpha \, I_\omega(z) \,.
\end{equation}

The Hamiltonian $H_z$ is concerned only with the particles in the
support of $\chi_z$; similarly for the phonon field,
$\rho(y)$ enters only for $y$ in this support. Moreover, $H_z$
commutes with $n_z= \sum_{i=1}^N \theta_z(x_i)$, the number of
particles in the support of $\chi_z$, where $\theta_z$ denotes the
characteristic function of the support of $\chi_z$. We can thus look
for a lower bound on $H_z$ in a fixed sector of $n_z$ particles. We
will prove the following lower bound.

\begin{Lemma} With $[t]_+ = \max\{t,0\}$,
\begin{equation}\label{hzlow}
H_z \geq -  \alpha  \left[ 4 \alpha n_z - \omega \right]_+ \sum_{i=1}^N \chi_z(x_i)^2 - \frac{3\alpha n_z}{(2\pi)^4} \left( \sqrt{\frac{2\pi}{3\omega}} \|\chi\|_\infty  +\|\nabla \chi\|_2  \right)^2 \,.
\end{equation}

\end{Lemma}

\begin{proof}
Pick some $\Lambda\geq \omega$. Applying \eqref{eq:posdef} with the positive definite function $f(x) = |x|^{-1}(e^{-\omega|x|} - e^{-\Lambda|x|})$, we have 
\begin{equation}
I_\omega(z) \geq  I_\Lambda(z) -  \frac 1 2  \left(\Lambda-\omega\right) \sum_{i=1}^N \chi_z(x_i)^2 \,.
\end{equation}
The last term represent the `self-energy' terms.
We will choose $\Lambda$ proportional to $n_z$, hence this term is of the order $n_z^2$. 
For the remaining terms, we complete the square and write
\begin{align}\nonumber
&   \chi_z(x_i)p_i^2 \chi_z(x_i) + p_i \chi_z(x_i)^2 p_i -   2\alpha  \chi_z(x_i) \int dy \frac{e^{-\Lambda|x_i-y|}}{|x_i-y|}\chi_z(y)\left(\rho(y)+\rho(y)^\dagger\right) \\ \nonumber &=  \left(  \chi_z(x_i) p_i - A_z(x_i)\right)\left(p_i\chi_z(x_i) - A_z(x_i)^\dagger\right) \\ \nonumber & \quad  +   \left(   p_i\chi_z(x_i) + A_z(x_i)^\dagger\right)\left(\chi_z(x_i)p_i + A_z(x_i)\right) \\ & \quad - 2 A_z(x_i)^\dagger A_z(x_i) - \left[ A_z(x_i), A_z(x_i)^\dagger \right] \,, \label{eq:comm}
\end{align}
where $A_z(x)$ is a vector operator with three components,
\begin{equation}
A_z(x) :=\frac \alpha{\pi^2} \int dy\, \chi_z(y)\rho(y) \int dk \frac{k\, e^{ik(y-x)}}{k^2\left(k^2+\Lambda^2\right)} \,.
\end{equation}
The Schwarz inequality applied to this last equation shows that 
\begin{equation}
A_z(x)^\dagger A_z(x) \leq  \int dw\, dy\, \chi_z(w)\rho(w)^\dagger \frac{e^{-\Lambda|w-y|}}{|w-y|}\rho(y)\chi_z(y) \,  \frac {2\alpha^2}{\pi^2} \int dk \frac {  1}{k^2(k^2+\Lambda^2)}\,.
\end{equation}
Moreover,
\begin{align}\nonumber
[A_z(x),A_z(x)^\dagger] 
&= \frac{2\alpha}{(2\pi)^4} \sum_{i=1}^3 \int \left| \nabla_y\left( \chi_z(y) f_i(\Lambda(x-y))\right)\right|^2 dy \\ 
& \leq  \frac{2\alpha}{(2\pi)^4} \sum_{i=1}^3 \left( \Lambda^{-1/2} \|\chi\|_\infty \|\nabla f_i\|_2 +\|\nabla \chi\|_2 \|f_i\|_\infty  \right)^2 \label{eq:commbound}
\end{align}
with 
\begin{equation}
f_i(x) = \frac {x_i}{|x|} \frac{1-(1+|x|)e^{-|x|}}{|x|^2} \quad , \quad x=(x_1,x_2,x_3) \,.
\end{equation}
We choose $\Lambda$ such that 
\begin{equation}
n_z \frac{2\alpha^2}{\pi^2} \int dk \frac {1}{k^2(k^2+\Lambda^2)} =  \frac{4\alpha^2 n_z}{\Lambda} \leq   \alpha \,,
\end{equation}
which then assures that the last two terms of equation \eqref{eq:comm} are relatively bounded by the field energy terms, i.e., the last term of $I_{\omega}(z)$ in \eqref{eq:iomega}. More precisely, with the choice $\Lambda= \max\{\omega, 4 \alpha n_z\}$ Lemma 4 follows, using the facts $\|f_i\|_\infty = 1$, $\|\nabla f_i\|_2^2= 2\pi/3$, and the commutator bound \eqref{eq:commbound}, this bound contributing the last term on the right side of \eqref{hzlow}.
\end{proof}

We now complete the proof of Theorem 3. If we insert bound (\ref{hzlow}) into (\ref{hhz}), we obtain the following lower bound on $H^{(N)}_U$:
\begin{align}\nonumber 
H^{(N)}_U & \geq -  \alpha  \int dz \left(\omega + \left[ 4\alpha n_z-\omega\right]_+\right) \sum_{i=1}^N \chi_z(x_i)^2 + \left(U - 2\alpha\right) V_C  \\ & \quad - \frac N 2 \int dx\,|\nabla\chi(x)|^2 -  \frac{3\alpha}{(2\pi)^4} N\, |{\rm supp \,}\chi| \left( \sqrt{\frac{2\pi}{3\omega}} \|\chi\|_\infty  +\|\nabla \chi\|_2  \right)^2\,.
\end{align}
The volume of the support of $\chi$, $|{\rm supp\,}\chi|=\int \theta_0$, enters via the identity $\int dz\, n_z = N |{\rm supp\,} \chi|$. 
We further bound $ \left[ 4\alpha n_z-\omega\right]_+ \leq 4\alpha n_z$, 
and use that 
\begin{align}\nonumber
\int dz \, n_z \sum_{i=1}^N \chi_z(x_i)^2  & =  \sum_{i,j=1}^N  \int dz \, \theta_z(x_j)  \chi_z(x_i)^2 \\ & =  2 \sum_{1\leq i< j \leq N}  \int dz \, \theta_z(x_j)  \chi_z(x_i)^2 +  N \,.
\end{align}
 Moreover, 
\begin{equation}
\int dz \, \theta_z(x_j)  \chi_z(x_i)^2 \leq \frac{Z}{|x_i-x_j|}
\end{equation}
with 
\begin{equation}
Z := \sup_{x\in\R^3} |x| \left(\theta_0*\chi^2\right)(x)\,.
\end{equation}
The final result is
\begin{align}\nonumber
H^{(N)}_U & \geq \left( U -  2\alpha - 8 \alpha^2 Z \right) V_C    - 4\alpha^2 N -   \alpha  N \omega  \\ & \quad - \frac N 2 \|\nabla\chi\|_2^2 -  \frac{3\alpha}{(2\pi)^4} N\, |{\rm supp \,}\chi| \left( \sqrt{\frac{2\pi}{3\omega}} \|\chi\|_\infty  +\|\nabla \chi\|_2  \right)^2\,. \label{3.21}
\end{align}
Note that $Z$ is bounded above by the diameter of the support of $\chi$, which can be chosen arbitrarily small. In particular, we can choose the diameter small enough such that $8\alpha^2 Z \leq U-2\alpha$, which leads to a lower bound on $H^{(N)}_U$ that is linear in $N$. This concludes the proof of Theorem~\ref{thm:stabalpha}.\qed

\bigskip

For $U = \nu \alpha$, $\nu > 2$, our lower bound is proportional to $\alpha^2 N$ 
for large $\alpha$. To see this, we choose the diameter of the support of $\chi$ 
to be of the order $1/\alpha$. Hence $Z\sim \alpha^{-1}$, and also $\omega \sim 
\alpha$ by scaling. Moreover, $\|\nabla\chi\|_2 \sim \alpha$, $\|\chi\|_\infty\sim 
\alpha^{3/2}$ and $|{\rm supp\, } \chi| \sim \alpha^{-3}$, hence the right side of (\ref{3.21}) is of the 
desired form, namely, $-\const \alpha^2 N$ for large $\alpha$.

We conjecture that, for $U = \nu \alpha$, $\nu>2$, 
\begin{equation}\label{conj:alpha}
H_U^{(N)}\geq N  E^{(1)}(\alpha) - C_\nu N \alpha^2  \quad \text{for all $\alpha>0$,}
\end{equation}
for some constant $C_\nu$ 
depending only on $\nu$. For $N=2$ this was proved in the previous section, but 
the proof of (\ref{conj:alpha}) for general $N\geq 2$ remains an \emph{open problem}.



\section{The $N$-polaron problem: Absence of Binding}\label{sec:nobindn}

We now return to the question of binding of polarons and prove Theorem~\ref{thm:nobindingN}.  Because of subadditivity of the energy (\ref{eq:subadd}), $E^{(N)}_U(\alpha) \leq N
E^{(1)}(\alpha)$ for any $N$, $U$ and $\alpha$. Hence it remains to
prove the reverse inequality. 

We perform a localization similar to that in the two-polaron case, but {\it
  relative to the nearest neighbor}. This type of localization is one
of the main technical ingredients in our proof. As in the bipolaron
case, the goal will be to localize each particle in a box whose size is
of the same order as the distance to the closest particles, as long as
this distance is not too small.

Let $\varphi_i$ be given as in (\ref{phi1})--(\ref{phi2}), for some $\ell>0$ and $b>1$. 
If $t_i$
denotes the distance of $x_i$ to the nearest neighbor among the $x_j$,
$j\neq i$, then
\begin{equation}
1 = \sum_{j_1,\dots,j_N}  \prod_{i=1}^N |\varphi_{j_i}(t_i)|^2
\end{equation}
and, by the IMS localization formula, 
\begin{align}\nonumber
&\langle \psi | H^{(N)}_U |\psi \rangle \\ & = \sum_{j_1,\dots,j_N} \left\langle \psi \prod_i \varphi_{j_i}(t_i) \left| H^{(N)}_U - \sum_{i=1}^N \sum_{j=1}^N \sum_{k}   \left|\nabla_i   \varphi_{k}(t_j)\right|^2 \right|  \psi \prod_i \varphi_{j_i}(t_i)  \right\rangle\,.
\end{align}
We claim that the following bound on the localization error holds.

\begin{Lemma}\label{lem:locbound}
On the support of $\varphi_{j_i}(t_i) $,
\begin{equation}
\sum_{j=1}^N \sum_{k}   \left|\nabla_i \varphi_{k}(t_j)\right|^2 \leq \frac {\gamma}{(\ell-\ell/b)^2} b^{2(1- j_i)}
\end{equation}
with $\gamma := 13 (\pi /2)^2$. 
\end{Lemma}

\begin{proof}
  Note that $\varphi_k(t_j)$ depends on $x_i$ in one of two ways.
  First, through $t_i$ when $j=i$, but also through all the $t_j$, $j\neq
  i$, where $x_i$ happens to be the nearest neighbor of $x_j$.

  We claim that there can be at most 12 of those $x_j$. If $x$ is the
  nearest neighbor of both $x_j$ and $x_k$, then $|x_j-x_k|\geq
  \max\{|x_j-x|,|x_k-x|\}$, and hence the angle between $x_j-x$ and
  $x_k-x$ is at least $\pi/3$. Think of $x$ as the center of a unit
  sphere. The lines from $x$ to each of these $x_j$'s intersects the
  unit sphere at certain points $p_j$, whose angular separation is at
  least $\pi/3$.  At each of these points $p_j$ we can, therefore,
  construct a unit sphere tangent at $p_j$ to the given sphere around
  $x$. From the packing problem we know there can be at most 12 such
  spheres. This proves the claim.

On the support of $\varphi_{j_i}(t_i)$, 
\begin{equation}
 \sum_k   \big |\nabla_i \varphi_{k}(t_i) \big|^2  \leq \frac{\pi^2}{4(\ell-\ell/b)^2} b^{2(1-j_i)}
\end{equation}
as we have already used in (\ref{eq:locerror}). If $x_i$ is the nearest neighbor of $x_j$, the same is true with $t_i$ replaced by $t_j$ on the left side, since $t_j \geq t_i$ by definition, and the left side is easily seen to be decreasing in $t_i$. This concludes the proof.
\end{proof}

We now proceed with the one-particle localization as in the
two-polaron case, localizing particle $i$ in a ball of radius $b^{j_i}
L$ centered at $u_i$, with $L< \ell/(4 b^2)$. More precisely, with $\chi$ given in (\ref{defchi}), let
\begin{equation}
\psi_{\bfj,\bfu}(X) = \psi(X) \prod_{i=1}^N \left[ \varphi_{j_i}(t_i) (b^{j_i}L)^{-3/2} \chi(b^{-j_i}(x_i-u_i)/L)\right]
\end{equation}
where $\bfj=(j_1,\dots,j_N)$ and $\bfu=(u_1,\dots,u_N)$. We have
\begin{equation}\label{norm}
\|\psi\|^2 = \sum_\bfj \int_{\R^{3N}} d\bfu\, \|\psi_{\bfj,\bfu}\|^2
\end{equation}
and, using Lemma~\ref{lem:locbound},
\begin{align}\nonumber
&\langle \psi|H^{(N)}_U|\psi\rangle \\ \nonumber &=  \sum_{\bfj} \int_{\R^{3N}} d\bfu  \left\langle \psi_{\bfj,\bfu} \left| H^{(N)} _U- \sum_{i=1}^N \left(\sum_{k_1,\dots,k_N}   \left|\nabla_i \mbox{$\prod_j$}  \varphi_{k_j}(t_j)\right|^2+ \frac{\|\nabla \chi\|_2^2}{b^{2 j_i}L^2} |\right)  \right|  \psi_{\bfj,\bfu}  \right\rangle \\ & \geq \sum_{\bfj} \int_{\R^{3N}} d\bfu  \left\langle \psi_{\bfj,\bfu} \left| H^{(N)}_U - \sum_{i=1}^N b^{-2j _i} \left( \frac {\gamma\,b^2}{(\ell-\ell/b)^2}+ \frac {\pi^2}{L^2} \right)  \right|  \psi_{\bfj,\bfu} \right\rangle \,.\label{fl}
\end{align}
In analogy with the two-particle
problem, the goal here is to show that the integrand in this
last expression is bounded below by $N E^{(1)}(\alpha) \|\psi_{\bfj,\bfu}\|^2$
which, together with (\ref{norm}), implies the conclusion of the theorem.

For given $\bfj$ and $\bfu$, let $B_i$ denote the ball of radius
$b^{j_i}L$ centered at $u_i$.  Because of our assumption
$L<\ell/(4b^2)$, the balls $B_i$ with $j_i\geq 1$ do not intersect any
of the other balls. Let $d_{ik}$ denote the distance between ball
$B_i$ and ball $B_k$. 

Recall that the ground state energy can be obtained from the $T\to
\infty$ asymptotics of the functional integral 
(\ref{eq:pathint2}).  In the case of relevance here, for states having
the aforementioned support properties of $\psi_{\bfj,\bfu}$, the
Brownian paths $\omega_i$ in the functional integral are confined to
the respective balls $B_i$. In addition to this confinement, the paths
have the property that at any given time $t$, the separation between
any $\omega_i(t)$ and its nearest neighbor among the $\omega_k(t)$,
$k\neq i$, satisfies the conditions according to the support of
$\varphi_{j_i}$. We may relabel the particles
such that $j_i=0$ for $i\leq M$, and $j_i\geq 1$ for $M<i\leq N$.  The
exponential in the functional integral is a sum of three pieces,
$\mathcal A+\mathcal B+\mathcal C$, where
\begin{equation}\label{eq:expp1}
\mathcal A = \sum_{k=M+1}^N \sum_{i=1}^{k-1}  \left( 2 \alpha \int_\R \frac{ds\, e^{-|s|}}{2} 
 \int_0^T  \frac {dt}{|\omega_i(t)-\omega_k(t+s)|} - U \int_0^T  \frac {dt}{|\omega_i(t)-\omega_k(t)|} \right)\,,
\end{equation}
\begin{equation}\label{eq:expp2}
\mathcal B = \sum_{k=M+1}^N   \alpha \int_\R \frac{ds\, e^{-|s|}}{2} 
 \int_0^T  \frac {dt}{|\omega_i(t)-\omega_i(t+s)|} \,,
\end{equation}
and
\begin{equation}\label{eq:expp3}
\mathcal C= \alpha \sum_{i,k=1}^M \int_\R \frac{ds\, e^{-|s|}}{2} 
 \int_0^T  \frac {dt}{|\omega_i(t)-\omega_k(t+s)|} 
- U \sum_{1\leq i<k\leq M} \int_0^T  \frac {dt}{|\omega_i(t)-\omega_k(t)|} \,.
\end{equation}

We first bound $\mathcal A$. For $k>M$ and $i\neq k$,  the distance $d_{ik}$ between the balls $B_i$ and $B_k$ is nonzero. Since the paths $\omega_i$ and $\omega_k$ are confined to the balls $B_i$ and $B_k$, respectively,  
$$
\mathcal A \leq T \sum_{k=M+1}^N \sum_{i=1}^{k-1}  \left( \frac{2 \alpha}{d_{ik}} - \frac U {d_{ik}+ 2L( b^{j_i} + b^{j_k}) }\right)\,.
$$
Similar to (\ref{bounddd}), 
\begin{equation}
d_{ik} \geq  b^{\max\{j_i,j_k\}-2}\ell - 4 b^{\max\{j_i,j_k\}} L \,,
\end{equation}
and hence 
\begin{equation}\label{eq:abound}
\mathcal A \leq - T \left( U \left( 1 -  \frac{4 L b^2}{\ell} \right)  -  2\alpha \right) \sum_{k=M+1}^N \sum_{i=1}^{k-1}  \frac 1{d_{ik}} \,.
\end{equation}
Under the assumption that  
$ U( 1 -  4 L b^2/\ell)>   2\alpha $
 this is negative. We not only want it to be negative, however, we also want it to dominate part of the localization error, namely,
$$
\left( \frac {\gamma\,b^2}{(\ell-\ell/b)^2}+ \frac {\pi^2}{L^2} \right)\sum_{k=M+1}^N b^{-2j _k} \,.
$$
Using the fact that $\min_{i\neq k} d_{ik} \leq \ell b^{j_k}\leq \ell b^{2j_k-1}$ we can bound 
\begin{equation}\label{eq:locbig}
 \sum_{k=M+1}^N b^{-2j _k}  \leq  2 \ell b^{-1} \sum_{k=M+1}^N \sum_{i=1}^{k-1}  \frac 1{d_{ik}} \,.
\end{equation}
From \eqref{eq:abound} and \eqref{eq:locbig}, we conclude that
\begin{equation}
 \label{eq:aest}
\mathcal A + T \left( \frac {\gamma\,b^2}{(\ell-\ell/b)^2}+ \frac {\pi^2}{L^2} \right)\sum_{k=M+1}^N b^{-2j _k} \leq 0
\end{equation}
as long as
\begin{equation}\label{U1}
\boxed{ \ 
  U \left( 1 -  \frac{4L b^2}{\ell} \right)  \geq  2\alpha + \frac{2\ell}b \left( \frac {\gamma\, b^2}{(\ell-\ell/b)^2}+ \frac {\pi^2}{L^2} \right) \,.  \ }
\end{equation}
Since we are seeking a lower bound on the energy, i.e., an upper bound on the functional integral, \eqref{U1} then implies that $\mathcal A$, together with the localization terms coming from $M< i\leq N$ in \eqref{fl}, are indeed negative by \eqref{eq:aest}. Therefore these terms can be discarded in the functional integral. This concludes the discussion of the term $\mathcal A$ and leaves us with $\mathcal B+\mathcal C$ and the remaining localization terms from \eqref{fl} corresponding to $1\leq i\leq M$.

Since $\mathcal B$ and $\mathcal C$ refer to different, now non-interacting, sets of particles (namely, $M+1\leq k\leq N$ and $1\leq i\leq M$), we see that the functional integral factorizes. The term $\mathcal B$ is just the exponent in the path integral for $(N-M)$ non-interacting polarons, each with its own field. The integral of $e^\mathcal B$ contributes $(N-M) E^{(1)}(\alpha)$ to the energy. Henceforth we can forget about $\mathcal B$.

With the aid of our previous linear lower bound for $U>2\alpha$, term $\mathcal C$ is almost as simple as term $\mathcal B$. Since $U>2\alpha$ we can write $U=2\alpha + Z + V$ with $Z$ and $V$, both positive, to be chosen later. Because of the separation condition for any $1\leq i\leq M$ and any time $t$, the distance between $\omega_i(t)$ and its nearest neighbor among the $\omega_k(t)$'s is bounded above by $ \ell$, and hence 
\begin{equation}
\sum_{1\leq i<k\leq M} \frac V {|\omega_i(t)-\omega_k(t)|} \geq  \frac V{ \ell}  M \,.
\end{equation}
The integral of $e^\mathcal C$ contributes at least
\begin{equation}\label{cond2}
E_{2\alpha+Z}^{(M)}(\alpha) + M \frac V \ell
\end{equation}
to the energy. By Theorem \ref{thm:stabalpha} this is bounded from below by $-M C(2\alpha+Z,\alpha) + M V/\ell$, where $C(2\alpha+Z,\alpha)$ is the finite constant implicit in Theorem \ref{thm:stabalpha}. This term is at least $$
M E^{(1)}(\alpha) + M \left( \frac {\gamma\,b^2}{(\ell-\ell/b)^2}+ \frac {\pi^2}{L^2} \right)
$$ 
(the second term being the localization error from \eqref{fl}) provided we take
\begin{equation}
V = \ell \left( E^{(1)}(\alpha) + \frac {\gamma\,b^2}{(\ell-\ell/b)^2}+ \frac {\pi^2}{L^2} + C(2\alpha+Z,\alpha) \right) \,.
\end{equation}
Another way to state this is that $U$ must satisfy
\begin{equation}\label{U2}
\boxed{ \   U \geq 2\alpha + Z + \ell\left( \frac {\gamma\,b^2}{(\ell-\ell/b)^2}+ \frac {\pi^2} {L^2} \right)  + \ell \sup_{n\geq 2} \left| \frac{E^{(n)}_{2\alpha+Z}(\alpha)}{n}- E^{(1)}(\alpha)  \right| \ } 
\end{equation}
for some $Z>0$.

We have thus shown that, for any given $\bfj$ and $\bfu$, the
integrand in the first line of (\ref{fl}) is bounded from below by $N
E^{(1)}(\alpha)\|\psi_{\bfj,\bfu}\|^2$ as long as $U$ satisfies the
bounds (\ref{U1}) and (\ref{U2}) (for some $Z>0$). In combination with
(\ref{norm}), this concludes the proof of Theorem~\ref{thm:nobindingN}.
\qed

\bigskip
There are many parameters in (\ref{U1}) and (\ref{U2}): $b$, $\ell$, $L$ and $Z$. The only constraint on them is $b>\ell>4Lb^2$, and each choice gives rise to a computable estimate on the critical $U$. We emphasize that the resulting bound on $U$ is independent of $N$. 

Our bound on the critical value $U_c(\alpha)$ is proportional to
$\alpha$ for large $\alpha$. This follows with the choice $\ell \sim L
\sim \alpha^{-1}$ and $b= O(1)$, in which case condition (\ref{U1}) is
of the form $U \geq \const \alpha$.  With $Z\sim \alpha$, condition
(\ref{U2}) is also of this form for large $\alpha$, since the last
term is bounded by $\const \alpha^2$ for large $\alpha$, as shown in
the previous section. We conjecture that the last term in (\ref{U2})
is actually bounded by $\alpha^2$ for \emph{all} $\alpha$, as
explained at the end of the previous section, Eq.~\eqref{conj:alpha}.
Assuming the validity of~\eqref{conj:alpha}, our method leads to the
bound $U_c(\alpha)\leq \const \alpha$ for \emph{all} $\alpha>0$.


\section{The Pekar-Tomasevich functional}
\label{sec:pekar}


\subsection{Boltzons for {\mathversion{bold}$U>2\alpha$}}

We shall prove the analogue of Theorem \ref{thm:stabalpha} for the PT functional. The designation `boltzons' refers to particles without any symmetry restriction.

\begin{proposition}\label{pekarstab}
 If $U>2\alpha$ then
\begin{equation}
 \label{eq:pekarstab}
\mathcal E^{(N)}_U(\alpha) \geq  
\begin{cases}
 - (0.0280) N U^3/(U-2\alpha) & \quad\text{if}\ 2\alpha<U<3\alpha \,,\\
- (0.755) N \alpha^2 & \quad\text{if}\ U\geq 3\alpha \,.
\end{cases}
\end{equation}
\end{proposition}

We note that this proves \eqref{eq:stabpekar} with the constant stated in \eqref{eq:stabpekarconst}.

\begin{proof}
We write $U=2(\alpha+\delta)$ with some $\delta>0$. Given any $\psi$, and hence $\rho_\psi$, we will use two inequalities to bound the first two terms in the functional \eqref{eq:pekar} in terms of $\rho_\psi$. The first is the Hoffmann-Ostenhof inequality \cite{HoHo} (see also \cite[Cor. 8.4]{LiSe}),
\begin{equation}
 \label{eq:lt2}
\sum_{i=1}^N \int_{\R^{3N}} |\nabla_i \psi(X)|^2 \,dX \geq \int_{\R^3} |\nabla \sqrt{\rho_\psi}(x)|^{2} \,dx \,.
\end{equation}
The second is the Lieb-Oxford inequality \cite{LiOx}, \cite[Thm. 6.1]{LiSe}
\begin{equation}
 \label{eq:lo}
\sum_{i<j} \int_{\R^{3N}} \frac{|\psi(X)|^2}{|x_i-x_j|} \,dX 
\geq \frac{1}{2} \iint_{\R^3\times\R^3} \frac{\rho_\psi(x)\,\rho_\psi(y)}{|x-y|} \,dx\,dy - (1.68) \int_{\R^3} \rho_\psi(x)^{4/3} \,dx \,.
\end{equation}
These two bounds imply that (with $\phi:=\sqrt{\rho_\psi/N}$)
\begin{equation}
 \label{eq:pekarstabproof}
\frac 1 N \mathcal P^{(N)}_{U}[\psi] \geq \int_{\R^3}\! \left( |\nabla \phi|^{2} - (1.68)U N^{1/3} \phi^{8/3}\right) dx
+ \delta N \iint_{\R^3\times\R^3} \frac{\phi(x)^2\,\phi(y)^2}{|x-y|} \,dx\,dy \,.
\end{equation}
Next, we use H\"older's inequality
\begin{equation}
\int \phi^{8/3} \,dx \leq \left( \int_{\R^3} \phi^3 \,dx \right)^{2/3} \left( \int_{\R^3} \phi^2 \,dx \right)^{1/3}
\end{equation}
and Lemma \ref{schwarz} below to conclude that
\begin{equation}
\int \phi^{8/3} \,dx \leq \frac 1{(4\pi)^{1/3}} \left( \iint \frac{\phi(x)^2\,\phi(y)^2}{|x-y|} \,dx\,dy \right)^{1/3} 
\left(\int |\nabla \phi|^{2} \,dx \right)^{1/3} \left( \int \phi^2 \,dx \right)^{1/3} .
\end{equation}
Using $(\alpha\beta\gamma)^{1/3} abc \leq \frac13 (\alpha a^3 +\beta b^3 +\gamma c^3)$ for non-negative numbers $a,b,c,\alpha,\beta,\gamma$, we see that
\begin{align}\nonumber 
 (1.68)U N^{1/3} \int \phi^{8/3} \,dx 
\leq & \int |\nabla\phi|^2 \,dx + \delta N \iint \frac{\phi(x)^2\,\phi(y)^2}{|x-y|} \,dx\,dy \\
& +\frac{(1.68 U)^3}{4\pi 3^3 \delta} \int \phi^2 \,dx \,.
\end{align}
This, together with \eqref{eq:pekarstabproof} leads to the lower bound
\begin{equation}
\mathcal E^{(N)}_U(\alpha) \geq -\frac{(1.68)^3}{54\pi} \frac{U^3}{U-2\alpha} N \,.
\end{equation}
While this is true for all $U>2\alpha$, the right side is not a monotone increasing function of $U$, which we know the left side to be. Therefore we can say that $\mathcal E^{(N)}_U(\alpha)$ is bounded from below by the maximum value of the right side once $U$ exceeds the maximum point, which is $U=3\alpha$. This concludes the proof.
\end{proof}

\begin{Lemma}\label{schwarz}
 For non-negative functions $\phi$
\begin{equation}
\left( \int_{\R^3} \phi^3 \,dx \right)^2 \leq \frac 1{4\pi} \iint_{\R^3\times\R^3} \frac{\phi(x)^2\,\phi(y)^2}{|x-y|} \,dx\,dy\  \int_{\R^3} |\nabla\phi(x)|^2 dx \,.
\end{equation}
\end{Lemma}

\begin{proof}
 We apply Schwarz's inequality
\begin{equation}
\left( \int_{\R^3} \phi^3 \,dx \right)^2  = \left\langle |p|^{-1} \phi^2 \left|\right. |p| \phi \right\rangle^2
\leq \left\langle \phi^2 \left| \frac{1}{p^2} \right| \phi^2 \right\rangle 
\left\langle \phi \Big| p^2 \Big| \phi \right\rangle
\end{equation}
and recall that $|p|^{-2}$ is convolution with $(4\pi|x|)^{-1}$.
\end{proof}

The stability of the PT functional with critical repulsion $U=2\alpha$ remains an \emph{open problem}. In the fermionic case the answer is affirmative, as was shown by Griesemer and M\o ller \cite{GrMo}. For the reader's convenience we include the proof here.

Combining the Lieb-Thirring inequality \cite{LiThi} \cite[Cor. 4.1]{LiSe}
\begin{equation}
 \label{eq:lt}
\sum_{i=1}^N \int_{\R^{3N}} |\nabla_i \psi(X)|^2 \,dX \geq K q^{-2/3} \int_{\R^3} \rho_\psi(x)^{5/3} \,dx \,.
\end{equation}
where $K:=\frac95 (2\pi)^{2/3}=6.13$ \cite{DoLaLo}, with the Lieb-Oxford inequality \eqref{eq:lo}, one deduces
\begin{equation}
\mathcal P^{(N)}_{2\alpha}[\psi] \geq K q^{-2/3} \int_{\R^3} \rho_\psi(x)^{5/3} \,dx - 2\alpha (1.68) \int_{\R^3} \rho_\psi(x)^{4/3} \,dx \,.
\end{equation}
The minimization of the expression on the right side under the normalization constraint $\int_{\R^3} \rho_\psi(x) \,dx=N$ leads to the lower bound
\begin{equation}
\mathcal P^{(N)}_{2\alpha}[\psi] \geq - \frac{(1.68)^2}{K} q^{2/3} N \alpha^2 \,,
\end{equation}
as claimed in Eq.~\eqref{eq:thm4}.


\subsection{Absence of binding in the Pekar-Tomasevich model}

We have just given an alternative proof of the fact that there is stability of matter in the PT model when $U>2\alpha$. Now we discuss the other part of Corollary \ref{pekarcor}, that is, the absence of binding for large $U$.

The first step is to linearize the problem. The variational problem for the PT functional can, equivalently, be written
\begin{equation}\label{eq:lin}
\mathcal E^{(N)}_U(\alpha) = \inf_{\psi,\sigma}\left\{  \left\langle \psi \left| \mathcal H^{(N)}_{U,\sigma} \right| \psi\right\rangle :\  \int_{\R^{3N}} |\psi|^2\,dX = 1 \right\} 
\end{equation}
where the $N$-particle Hamiltonian $\mathcal H^{(N)}_{U,\sigma} $ is defined to be
\begin{align}\nonumber
\mathcal H^{(N)}_{U,\sigma} &:= \sum_{i=1}^N  \left( -\Delta_i - 2\alpha \int_{\R^3} \frac{\sigma(y)}{|x_i-y|} \, dy \right) + U \sum_{i<j} \frac 1{|x_i-x_j|}   \\ & \quad  + \alpha \iint_{\R^3\times\R^3} \!\!\!\frac{\sigma(x)\, \sigma(y)}{|x-y|} \,dx\,dy 
\,. \label{eq:linham}
\end{align}
We proceed as before, by localizing particles into individual boxes,
with sizes depending on the distance to the nearest neighbor. In each
localization region we obtain a lower bound on the energy of a given
$\psi$ by minimizing over $\sigma$, which yields the PT functional for the localized $\psi$. (In other words, we linearize, localize, and de-linearize. If we had not followed this route and tried to deal with the quartic term directly, the resulting expressions would be much more complicated.) 

Consider first the case $N=2$. In the region $j=0$, we need a lower bound on $\mathcal E_0^{(2)}(\alpha)$. We could use Lemma~\ref{lem2} above, but it is simpler, and indeed more accurate, to use
\begin{equation}
\mathcal E_0^{(2)}(\alpha) =  2\, \mathcal E_0^{(1)}(2\alpha) = 8\, \mathcal E_0^{(1)}(\alpha) = 2\, \mathcal E_0^{(1)}(\alpha) - 6 \cdot (0.109) \, \alpha^2\,.
\end{equation}
The first equality follows from the linearization \eqref{eq:lin} since the ground state of \eqref{eq:linham} for $U=0$ is a product function for every $\sigma$; the second equality follows by scaling.

For the regions $j\geq 1$, we obtain the PT functional for 2 particles
localized in disjoint balls. The proof of the corresponding lower
bound to the energy, analogous to Lemma~\ref{lem3}, is obvious,
bounding the attractive energy using the smallest possible distance of
the particles. Alternatively, one can use Lemma~\ref{lem4}. For $j\geq
1$ the resulting condition on $U$ is thus the same as in the proof of
Theorem~\ref{thm2}. Our improved estimate in the $j=0$ region leads to
the bound $\nu_c \leq 29.4$ (to be compared with the bound $C<2\cdot (26.6)=53.2$ for the Fr\"ohlich polaron).

We can similarly analyze the $N$-particle problem. For the particles
with $j\geq 1$ the bounds are exactly the same as before, except that
functional integrals are not needed in the derivation.  For the
particles with $j=0$ the improved stability bound (\ref{eq:pekarstab})
(or (\ref{eq:thm4}) for fermions) is used, and hence the final
condition for the absence of binding will be a lower value of $U$ than that
for the Fr\"ohlich Hamiltonian.\hfill\qed


\appendix
\section{Proof of Lemma~\ref{lem4}}

We assume that the confining balls $B_1$, $B_2$ are
each of radius $R$, and that the balls are of distance $d=
\inf\{|x_1-x_2|: x_1\in B_1,\,\,x_2\in B_2\}$ from each other.  In the
following, let
\begin{equation}
      \hat{a}(x)= \frac{1}{(2\pi)^{3/2}}\int dk\, e^{ikx}a(k) \, ,\  \hat{a}^{\dagger}(x)= \frac{1}{(2\pi)^{3/2}}\int dk\, e^{-ikx}a^{\dagger}(k)
\end{equation}
be normalized annihilation and creation operators, with $x$ 
a point in configuration space.  In terms of these operators, the particle-
field interaction term in $H^{(2)}_U$ is given by
\begin{equation}
  \sqrt{\alpha}\phi(x) = \frac{\sqrt{\alpha}}{\pi^{3/2}}\int dy\, \frac{\hat{a}^{\dagger}(y)+\hat{a}(y)}{|x-y|^2}
\end{equation}
and the phonon field energy by
   \begin{equation}
        H_f= \int dx\,\hat{a}^{\dagger}(x)\hat{a}(x).
\end{equation}

Fix a plane midway between the two balls and perpendicular
to the line between their centers,  and
let $S_1$ and $S_2$ be the resulting half-spaces with $B_1\subset S_1$ and
$B_2\subset S_2$. Then we have 
the identity
 \begin{align}\nonumber 
H^{(2)}_{U=0}  = & \ p^2_1-\frac{\sqrt{\alpha}}{\pi^{3/2}}\int_{S_1}dy\, 
\frac{\hat{a}^{\dagger}(y)+\hat{a}(y)}{|x_1-y|^2}\\
& +
\int_{S_1}dy \left(\hat{a}^{\dagger}(y)-\frac{\sqrt{\alpha}}{\pi^{3/2}|x_2-y|^2}\right)
\left(\hat{a}(y)-\frac{\sqrt{\alpha}}{\pi^{3/2}|x_2-y|^2}\right)\nonumber\\
&+ p^2_2-\frac{\sqrt{\alpha}}{\pi^{3/2}}\int_{S_2}dy \,\frac{\hat{a}^{\dagger}(y)+\hat{a}(y)}{|x_2-y|^2}\nonumber\\
& +\int_{S_2}dy \left(\hat{a}^{\dagger}(y)-\frac{\sqrt{\alpha}}{\pi^{3/2}|x_1-y|^2}\right)
\left(\hat{a}(y)-\frac{\sqrt{\alpha}}{\pi^{3/2}|x_1-y|^2}\right)\nonumber\\
& -\frac{\alpha}{\pi^3}\int_{S_1} \frac{dy}{|x_2-y|^4}-\frac{\alpha}{\pi^3}\int_{S_2}\frac{dy}{|x_1-y|^4}\, . \label{H3.eq}
\end{align}

Define 
   \begin{equation}
   \hat{a}_{x_{2}}(y) \equiv \left(\hat{a}(y)-\frac{\sqrt{\alpha}}{\pi^{3/2}|x_2-y|^2}\right),\,\,y \in S_1
\end{equation}
and analogous expressions for  $\hat{a}^{\dagger}_{x_{2}}(y)$,  and
for $\hat{a}_{x_{1}}(y), \hat{a}^{\dagger}_{x_{1}}(y)$ with $y\in S_2$. In terms of these operators, the identity (\ref{H3.eq})
becomes 
 \begin{align} \nonumber 
    H^{(2)}_{U=0}= & \ p^2_1-\frac{\sqrt{\alpha}}{\pi^{3/2}}\int_{S_1}dy 
    \,\frac{\hat{a}^{\dagger}_{x_{2}}(y)+\hat{a}_{x_{2}}(y)}{|x_1-y|^2}+\int_{S_1}dy\, \hat{a}^{\dagger}_{x_{2}}(y)\hat{a}_{x_{2}}(y) \nonumber\\
& + p^2_2-\frac{\sqrt{\alpha}}{\pi^{3/2}}\int_{S_2}dy\, \frac{\hat{a}^{\dagger}_{x_{1}}(y)+\hat{a}_{x_1}(y))}{|x_2-y|^2} +
\int_{S_2}dy\, \hat{a}^{\dagger}_{x_{1}}(y)\hat{a}_{x_{1}}(y)\nonumber\\
& -\frac{\alpha}{\pi^2|x_1\cdot n|}-\frac{\alpha}{\pi^2|x_2\cdot n|}-\frac{2\alpha}{|x_1-x_2|}\,, \label{app1.eq}
\end{align}
where $|x_i\cdot n|$ is the distance between $x_i$ and the dividing plane, $i=1,2$.
(The integrals in the last line of Eq.~(\ref{H3.eq}) can be done explicitly via
cylindrical coordinates, resulting in the first and second terms in
the last line of this identity; the last term is an integral 
of $|x_1-y|^{-2}|x_2-y|^{-2}$ over {\it all} of $\R^3$ and is readily computed
to be the Coulomb attraction term.)

We can give a lower bound on expectations of the right side of the first line of this last equation
(\ref{app1.eq}), assuming the first particle is indeed confined in
the ball $B_1$. Let $K_{S_{1},x_{2}}$ be the one-particle operator of the first line,
  \begin{equation}
    K_{S_{1},x_{2}}= p^2_1-\frac{\sqrt{\alpha}}{\pi^{3/2}}\int_{S_1}dy\, \frac{\hat{a}^{\dagger}_{x_{2}}(y)+\hat{a}_{x_{2}}(y)}{|x_1-y|^2}
     +\int_{S_1}dy \,\hat{a}^{\dagger}_{x_{2}}(y)\hat{a}_{x_{2}}(y) \,,
 \end{equation}
 which we regard as acting in the Hilbert space $L^{2}(S_{1})\otimes {\mathcal F}_{S_{1}}$, the
 latter factor being the Fock space associated with the phonon variables $y\in S_{1}$; the operator is 
 a function of $x_{2}\in S_{2}$.  Note that $p_{1}$ commutes with  $\hat{a}_{x_{2}}(y)$ and its adjoint 
 and that $\hat{a}_{x_{2}}(y)$ and  $\hat{a}_{x_{2}}^{\dagger}(y)$ satisfy the canonical commutation relations. The
 operator function $K_{S_{2},x_{1}}$ is defined analogously.
  
 Fix $x_{2}$, let $\psi$ be a state in $L^{2}(\R^3)\otimes {\mathcal F}_{S_{1}}$
supported in $x_1\in B_1$, 
and then consider a product state $\Psi= \psi\otimes\Phi\in L^2(\R^3)\otimes {\mathcal F}$ where $\Phi$ is
a coherent state of the phonon variables corresponding to $y\in S_2$ such
that 
\begin{equation}
        \hat{a}(y)|\Phi\rangle = \frac{\sqrt{\alpha}}{\pi^{3/2} |x_c-y|^2}\,|\Phi\rangle\, ,\ y\in S_2. 
\end{equation}
Here, we take $x_c\in S_{1}$ to be on the line passing through the centers of
the two balls and of distance $d/2+2R$ from the dividing plane (i.e.
as remote from the dividing plane as possible but on the surface of $B_1$).
 For
such a state $\Psi$, we  have that
\begin{align} \nonumber 
  E^{(1)}(\alpha) & \leq \langle \Psi | H^{(1)} |\Psi\rangle \\
  &= \langle \psi | K_{S_{1},x_{2}} |\psi\rangle -\frac{2\alpha}{\pi^3}
\int_{S_2} dy\,\Big\langle \psi\Big|\frac{1}{|x_1-y|^2|x_c-y|^2}\Big| \psi\Big\rangle  + \frac{\alpha}{\pi^3}\int_{S_2} \frac{dy}{|x_c-y|^4}\nonumber\\
    &\leq  \langle \psi | K_{S_{1},x_{2}}| \psi\rangle -\frac{2\alpha}{\pi^3}
 \inf_{x_1\in B_1}\int \frac{dy}{|x_1-y|^2|x_c-y|^2} +\frac{\alpha}{\pi^3(d/2+2R)}\,. \label{A2.eq}
\end{align}
The integral in the infimum is seen to have no critical points for
$x_1$ in the interior of $B_1$ and so attains its minimum for $x_1$ on the
boundary of $B_1$. The integral can again be written using cylindrical coordinates and
the angular integration performed explicitly. One then writes the integrand of the resulting double integral just as a function
of $z_1$, say, where $x_1= (r_1,\theta_1,z_1)$ in cylindrical coordinates, and where $(z_1+d/2+R)^2+r_1^2= R^2$.
Minimization of the integrand in this double integral regarded as a function of $z_1$ is
tedious but straightforward, the minimum occurring at $x_1= x_c$. The
integral $\int_{S_2} dy\, |x_c-y|^{-4}$ is equal to $\pi/(d/2+2R)$ as computed above for Eq.~(\ref{H3.eq}). Thus, we obtain
\begin{equation}\label{A3.eq}
 E^{(1)}(\alpha) \leq \langle \psi_1 | K_{S_{1},x_{2}} |\psi_1\rangle -\frac{\alpha}{\pi^2(d/2+2R)}.
\end{equation}
Of course the second line of Eq.~(\ref{app1.eq}) is handled similarly.

By this last inequality (\ref{A3.eq}) and Eq.~(\ref{app1.eq}), we have
that for any state $\Psi$ with electron support in $B_1\times B_2$, 
\begin{align}\nonumber 
  \langle\Psi | H^{(2)}_{U=0}| \Psi\rangle &= \langle \Psi| K_{S_{1},x_{2}}\otimes 1 | \Psi\rangle 
+\langle \Psi | K_{S_{2},x_{1}}\otimes 1 | \Psi\rangle  \\
 & \quad   -\alpha \Big{\langle}\Psi \Big| \left(\frac{1}{\pi^{2}|x_1\cdot n|}+ \frac{1}{\pi^{2} |x_2\cdot n|}
+\frac{2}{|x_{1}-x_{2}|}\right) \Big| \Psi\Big\rangle \nonumber\\
&\geq  2E^{(1)}(\alpha) +\frac{2\alpha}{\pi^{2}(d/2 +2R)} \nonumber\\
  & \quad   -\alpha \Big\langle\Psi \Big| \left(\frac{1}{\pi^{2}|x_1\cdot n|}+ \frac{1}{\pi^{2} |x_2\cdot n|}
+\frac{2}{|x_{1}-x_{2}|}\right) \Big| \Psi\Big\rangle \,.
\end{align} 
Noting that $|x_{1}\cdot n|$ and $|x_{2}\cdot n|$ are at least $d/2$, we
have that
\begin{equation}     
\langle\Psi | H^{(2)}_{U=0}| \Psi\rangle \geq  2E^{(1)}(\alpha) -2\alpha\Big\langle \Psi \Big| \frac{1}{|x_{1}-x_{2}|}\Big| \Psi\Big\rangle -\frac{16\alpha R}{\pi^{2}d
(d+4R)}\,,
\end{equation}
which is the claim of Lemma~\ref{lem4}.  \hfill \qed

\bigskip {\it Acknowledgments.} We are grateful to Herbert Spohn for making us aware of the problem of proving absence of binding
for large $U$. We also thank Marcel Griesemer and Jacob Schach M\o ller for helpful comments on an early version of our manuscript.
 Partial financial support from the U.S.~National
Science Foundation through grants PHY-0652854 (E.L. and R.F.) and
PHY-0845292 (R.S.) are gratefully acknowledged. L.T. would like to
thank the PIMS Institute, University of British Columbia, for their 
hospitality and support.


\end{document}